\documentclass[sort&compress,numbers,times,3p,fleqn,letter]{elsarticle}
\usepackage[utf8]{inputenc}
\usepackage{algorithm}
\usepackage{algorithmic}
\usepackage{amsmath,amssymb,amsthm,amsfonts}
\usepackage{enumitem}
\usepackage[dvipsnames]{xcolor}
\usepackage[colorlinks=true,linkcolor=RoyalBlue,citecolor=PineGreen,urlcolor=blue]{hyperref}
\usepackage{bm}
\newtheorem{thm}{Theorem}

\newtheorem{lem}[thm]{Lemma}

\newtheorem{prop}[thm]{Proposition}

\newdefinition{defn}[thm]{Definition}
\newdefinition{exmp}[thm]{Example}
\newdefinition{rem}[thm]{Remark}
\newdefinition{prob}[thm]{Problem}
\newdefinition{prin}[thm]{Principle}
\newdefinition{alg}{Algorithm}
\newdefinition{notation}[thm]{Notation}

\DeclareMathOperator{\rank}{rank}
\DeclareMathOperator{\Const}{Const}

\DeclareMathOperator{\Wrn}{Wr}
\DeclareMathOperator{\lead}{lead}

\newcommand{\Kdiff}{K^{\operatorname{diff}}}
\newcommand{\kacl}{k^{\operatorname{acl}}}

\newcommand{\C}{\mathbb C}

\begin{document}
\setlength{\abovedisplayskip}{6pt}
\setlength{\belowdisplayskip}{5pt}
\setlength{\abovedisplayshortskip}{0pt}
\setlength{\belowdisplayshortskip}{0pt}
\begin{frontmatter}
\title{Computing all identifiable functions  of parameters for ODE models}
\author[AO]{Alexey Ovchinnikov}\ead{aovchinnikov@qc.cuny.edu}
\author[AP]{Anand Pillay}\ead{Anand.Pillay.3@nd.edu}
\author[GP]{Gleb Pogudin}\ead{gleb.pogudin@polytechnique.edu}
\author[TS]{Thomas Scanlon}\ead{scanlon@math.berkeley.edu}
\address[AO]{CUNY Queens College, Department of Mathematics,
65-30 Kissena Blvd, Queens, NY 11367, USA \\ 
CUNY Graduate Center, Mathematics and Computer Science, 365 Fifth Avenue,
New York, NY 10016, USA}
\address[AP]{University of Notre Dame, Department of Mathematics, Notre Dame, IN 46556, USA}
\address[GP]{LIX, CNRS, \'Ecole Polytechnique, Institute Polytechnique de Paris, 1 rue Honor\'e d'Estienne d'Orves, 91120, Palaiseau, France\fnref{otheraddress}}
\address[TS]{University of California, Berkeley, Department of Mathematics, Evans Hall, Berkeley, CA 94720-3840, USA} 
\fntext[otheraddress]{During the preparation of this manuscript, G. Pogudin also worked at the Courant Institute of Mathematical Sciences, New York University, New York, NY 10012, USA and at the Higher School of Economics, Faculty of Computer Science, Moscow, 109028, Russia}

\begin{abstract}
Parameter identifiability is a structural property of an ODE model for recovering the values of parameters from the data (i.e., from the input and output variables). 
This property is a prerequisite for meaningful parameter identification in practice.
In the presence of nonidentifiability, it is important to find all functions of the parameters that are identifiable.
The existing algorithms check whether a given function of parameters is identifiable or, under the solvability condition, find all identifiable functions.  However, this solvability condition is not always satisfied, which presents a challenge.
Our first main result is an algorithm that computes all identifiable functions without any additional assumptions, which is the first such algorithm as far as we know.
Our second main result concerns the identifiability from multiple experiments (with generically different inputs and initial conditions among the experiments).
For this problem, we
 prove
that the set of functions identifiable from multiple experiments is what would actually  be computed by input-output equation-based  algorithms  
(whether or not the solvability condition is 
fulfilled), which was not known before. 
We give an algorithm that not only finds these functions but also provides an upper bound for the number of experiments to be performed to identify these functions.
We provide an implementation of the presented algorithms.
\end{abstract}
\begin{keyword}
parameter identifiability\sep multiple  experiments\sep input-output equations\sep differential algebra \sep characteristic sets
\end{keyword}
\end{frontmatter}

\section{Introduction}
In this paper, we study structural parameter identifiability of rational ODE systems.
Roughly speaking, a parameter is structurally identifiable if its value can be recovered from the observations assuming continuous noise-free measurements and sufficiently exciting inputs (also referred to as the persistence of excitation, see~\cite{LG94, V18}).
If not all of the parameters of a model are identifiable, the next question usually is what  rational
functions $h(\bar{\mu}) \in \C(\bar{\mu})$ of the parameters $\bar{\mu}$ are identifiable.
The knowledge of identifiable functions can be used in 
these
ways:
\begin{itemize}[itemsep=-0.05in,topsep=0.03in]
\item 
If the functions of interest to the modeler are identifiable, then the lack of identifiability of some parameters is not an issue (sometimes, this is even an advantage~\cite{S17}).
\item 
Identifiable functions can be used to find an identifiable reparametrization of the model~\cite{MED2009, MS2014, BD16}, which is a way of improving the model.

\item Knowledge of identifiable functions can be used to  discover 
parameter transformations that preserve the input-output behavior  and thus could provide additional insights to the modeler (see Section~\ref{subsec:sf}).

\end{itemize}

To the best of our knowledge, all existing approaches to computing identifiable functions extract them from the coefficients of input-output equations (going back to~\cite{OllivierPhD}; for a concise summary, we refer to~\cite[Introduction and Algorithm~II.1]{OPT19}).
To conclude that the coefficients of an input-output equation are identifiable, one can, for example, verify if the solvability condition~\cite[Remark~3]{Saccomani2003} holds for the equation.
The condition can be checked by an algorithm (see \cite[Section~4.1]{DVJBNP01} and~\cite[Section~3.4]{MXPW11}) and holds for some classes of models~\cite{OPT19}. 
If the condition does not hold, then this approach of finding identifiable functions of parameters is not applicable  but is still used by some of the existing software packages, including DAISY and COMBOS.
This is a reason why these tools may miss the non-identifiability of some of the parameters in such systems.
For a simple example of a system for which this condition does not hold, see~\cite[Example~2.14]{HOPY2018} (see also Sections~\ref{subsec:sf}  and~\ref{sec:LV-mult}). 

Therefore, we are not aware of any prior algorithm that can compute all identifiable functions (e.g., by computing generators of the field of identifiable functions).
Note that  
 some
existing software 
 can, for any fixed rational function of parameters,  
check whether it is identifiable~\cite[Remark~1]{ident-compare}.  However, looking for all identifiable functions, it is not known in advance what functions of parameters to test for identifiability, so this approach cannot be used as an algorithm.

 The above issues motivate
the following  
questions (which remained unanswered as far as we know) we study in this paper:
\begin{enumerate}[label = (Q\arabic*),itemsep=-0.05in,topsep=0.03in]
    \item\label{q:single} How to find the identifiable functions of a model even if the solvability condition does not hold?
    \item\label{q:multiple} If the solvability condition does not hold, what is the meaning of the coefficients of the input-output equations?
\end{enumerate}
Our main results are the following answers to these questions:
\begin{itemize}[itemsep=-0.02in,topsep=0.02in]
    \item We answer~\ref{q:single} by providing  Algorithm~\ref{alg:main} for computing  generators of the field of identifiable functions.
     This is the first such
    algorithm  and is based on 
     our theory established in Theorem~~\ref{prop:charset}.
    
    \item 
    We show in 
    Theorem~\ref{thm:me} 
    that the coefficients of the input-output equations are  generators of the field of functions identifiable \emph{from multiple experiments} (with generically different inputs and initial conditions  among the experiments~\cite{GenSSI2}), thus answering~\ref{q:multiple}.
    To the best of our knowledge, this natural interpretation of the coefficients of the input-output equations has not been known before despite the popularity of this method.
    Furthermore, we use this to derive 
    the first
    upper bound for the number of  such experiments, which can be used further for experimental design (e.g., for protocols such as~\cite[Section~7]{FN09}).
    The multi-experiment setup is natural, for example, for models involving constant inputs~\cite{VECB}.
\end{itemize}

The theoretical basis for this work  
 uses differential algebra and commutative algebra. 
We employ characteristic sets, a tool from computational algebra. 
The key difference with the prior algorithms based on  characteristic sets is that we provide a mathematically sound way to treat a typically ignored case in which the solvability condition is not satisfied.
To achieve this, we analyze the Wronskians of the monomials of characteristic sets using methods from linear algebra.
Our results are informed by model theory in the sense of mathematical logic, though this does not appear explicitly in our presentation. We  
 elaborate on
this connection in a follow-up work~\cite{second_paper}.
Additional related results on identifiability using input-output equations and differential algebra include  \cite{JNB2019, Saccomani2003,Meshkat14res,Meshkat18,DAISY,DAISY_IFAC,DAISY_MED,COMBOS,MRS2016,intdiff}.

The rest of the paper is organized as follows.
Section~\ref{sec:prelim} contains  definitions and notation that we use.
In Section~\ref{sec:ce}, we give our algorithm for computing the generators of the field of identifiable functions,  which is based on the theory we present in this section as well.
Section~\ref{sec:me} is on  theory for multi-experimental identifiability.
We illustrate our methods with examples in Section~\ref{sec:examples}.  We prove our main results in~\ref{sec:pfs}. In the other remaining appendices, we present and prove correctness of two algorithms that are used in our main algorithmic contributions and also provide a mathematical discussion, illustrated with examples, on our main theorems.

We have implemented Algorithm~\ref{alg:main} and an algorithm for computing the bound from Theorem~\ref{thm:me} (as in Remark~\ref{rem:alg_bound}) in Maple.
A {\sc Maple} implementation together with the examples from Section~\ref{sec:examples} is available at~\url{https://github.com/pogudingleb/AllIdentifiableFunctions}.  This implementation has recently been incorporated into a freely available web app~\url{https://maple.cloud/app/5710317752942592/SIAN}.

%%%%%%%%%%%%%%%%%%%%%%%%%%%%%%%%%%%%%%%%%%%%%%%%%%%%%%%%%%%%%%%%%%%%%%%%

\section{Basic notions and notation}\label{sec:prelim}
In this section, we will present the basic notions and notation from differential algebra and parameter identifiability that are essential for our main results.  

\subsection{Background and notation from differential algebra}
Differential algebra has been a standard theory behind identifiability, and we will simply fix the basic notation. General references include \cite{Kol,Ritt}.
For other presentations of these concepts in the context of control theory, see~\cite{LG94, Jiafan2009, FJ, Saccomani2003}.
\begin{notation}[Differential rings and ideals]\label{not:diffrings}
\begin{enumerate}[label = (\alph*),itemsep=-0.04in]
\item[]
\item
A {\em differential ring} $(R,\,')$ is a commutative ring with a derivation $'\!\!:R\to R$, that is, a map such that, for all $a,b\in R$, $(a+b)'=a'+b'$ and $(ab)'=a'b+ab'$. 
 A {\em differential field} is a differential ring that is a field.
For  $i>0$,  $a^{(i)}$ denotes the $i$-th order derivative of $a \in R$.
 $\Const(K)$ denotes the field of constants of a differential field $K$.
\item The {\em ring of differential polynomials} in the variables $z_1,\ldots,z_n$ over a differential field $(K,\,')$ is the ring \[K\big[z_j^{(i)}\mid i\geqslant 0,\, 1\leqslant j\leqslant n\big]\] with a derivation defined on the ring by $\big(z_j^{(i)}\big)' := z_j^{(i+1)}$. This differential ring is denoted by $K\{z_1,\ldots,z_n\}$. 
\item For differential fields $F\subset L$ and $a_1,\ldots,a_n \in L$, the smallest differential subfield of $L$ that contains $F$ and $a_1,\ldots,a_n$ is denoted by $F\langle a_1,\ldots,a_n\rangle$.
\item For a commutative ring $R$ and a subset $F\subset R$, the smallest ideal containing $F$ is denoted by $(F)$.
\item An ideal $I$ of a differential ring $(R,\,')$ is called a {\em differential ideal} if, for all $a \in I$, $a'\in I$. For $F\subset R$, the smallest differential ideal containing $F$ is denoted by $[F]$. 
  \item 
  For an ideal $I$ and element $a$ in a ring $R$, we denote $I \colon a^\infty = \{r \in R \mid \exists n\colon a^nr \in I\}$.
  This set is also an ideal in $R$. 
This will be useful for dealing with ODE systems in which  (non-polynomial) rational functions appear.
  \item\label{def:Wr} For $a_1,\ldots,a_n$ in a differential ring $R$, we denote the 
  $n\times n$ matrix with $(i,j)$-entry $a_j^{(i-1)}$ by $\Wrn(a_1,\ldots,a_n)$ and call it the {\em Wronskian} of $a_1,\ldots,a_n$. For example, \[\Wrn(a_1,a_2) =\begin{pmatrix}a_1&a_2\\  a_1'& a_2'\end{pmatrix}.\]

\end{enumerate}
\end{notation}

The rest of the definitions in this section generalize Gaussian elimination to systems of non-linear ODEs.
Differential rankings are analogous to ordering of variables in Gaussian elimination; characteristic sets and
presentations are analogous to row echelon form and reduced row echelon forms, respectively.

\begin{defn}\label{def:diffranking}
  A {\em differential ranking} is a total order $>$ on $Z := \{z_j^{(i)}\mid i\geqslant 0,\, 1\leqslant j\leqslant n\}$ satisfying:
    \[\forall\ x \in Z\ \ x' > x\quad \text{and}\quad
     \forall\ x, y \in Z\ \ (x >y\implies x' > y'). \]
\end{defn}

\begin{notation}\label{not:rank}
  For  $f \in K\{z_1,\ldots,z_n\} \backslash K$ and a differential ranking,
  \begin{itemize}[itemsep=-0.02in,topsep=0.01in]
    \item $\lead(f)$ is the element of $\big\{z_j^{(i)} \mid i \geqslant 0, 1 \leqslant j \leqslant n\big\}$ of the highest rank appearing in $f$. 
    This is partly analogous to the leading variable in Gaussian elimination.
    \item The leading coefficient of $f$ viewed as a polynomial in $\lead(f)$ is 
    called the {\em initial} of $f$. 
    This is similar to the leading coefficient in Gaussian elimination.
    \item The {\em separant} of $f$ is $\frac{\partial f}{\partial\lead(f)}$. 
   One can show that it is equal to the leading coefficient of any derivative of $f$.
    \item The {\em rank} of $f$ is $\rank(f) = \lead(f)^{\deg_{\lead(f)}f}$. 
    The ranks are compared first by $\lead$, and in the case of equality, by $\deg$.
    This is analogous to the leading variable in Gaussian  elimination/leading term in Gr\"obner bases.
    \item For $S \subset K\{z_1,\ldots,z_n\} \backslash K$, 
   the product
    of initials and separants of $S$ is denoted by $H_S$. 
   This is used in handling division with remainder algebraically.
  \end{itemize}
\end{notation}
\begin{defn}[Characteristic sets]
  \begin{itemize}[itemsep=-0.02in]
  \item[]
    \item For $f, g \in K\{z_1,\ldots,z_n\} \backslash K$, $f$ is said to be {\em reduced} w.r.t. $g$ if no proper derivative of $\lead(g)$ appears in $f$ and $\deg_{\lead(g)}f <\deg_{\lead(g)}g$.
    \item 
    A subset $\mathcal{A}\subset K\{z_1,\ldots,z_n\} \backslash K$
    is called {\em autoreduced} if, for all $p \in \mathcal{A}$, $p$ is reduced w.r.t. every  element of $\mathcal A\setminus \{p\}$. 
    Every autoreduced set is finite \cite[Section~I.9]{Kol}. 
    \item Let $\mathcal{A}=A_1<\ldots<A_r$ and $\mathcal{B} = B_1<\ldots<B_s$ be autoreduced sets ordered by their ranks (see Notation~\ref{not:rank}). We say that $\mathcal{A} < \mathcal{B}$ if
    \begin{itemize}[itemsep=0in,topsep=0in]
      \item $r > s$ and $\rank(A_i)=\rank(B_i)$, $1\leqslant i\leqslant s$, or
      \item there exists $q$ such that $\rank(A_q) <\rank(B_q)$ and, for all $i$, $1\leqslant i< q$, $\rank(A_i)=\rank(B_i)$.
    \end{itemize}
    \item An autoreduced subset of the smallest rank of a differential ideal $I\subset K\{z_1,\ldots,z_n\}$
    is called a {\em characteristic set} of $I$. One can show that every non-zero differential ideal in $K\{z_1,\ldots,z_n\}$ has a characteristic set.
  \end{itemize}
\end{defn}
\begin{defn}[Characteristic presentation] \label{def:char_pres} (cf. {\cite[Definition~3]{Boulier2000}})
 A polynomial is said to be {\em monic} if at least one of its coefficients is $1$. This is how monic is typically used in identifiability analysis and not how it is used in~\cite{Boulier2000}. A set of polynomials is said to be monic if each polynomial in the set is monic.
 
 Let $\mathcal C$ be a monic characteristic set of a prime differential ideal $P \subset K\{z_1,\ldots,z_n\}$.
 Let $N(\mathcal{C})$ denote the set of non-leading variables of $\mathcal{C}$.
    Then $\mathcal{C}$ is called a {\em characteristic presentation} of $P$ if all initials of $\mathcal C$ belong to $K[N(\mathcal C)]$  and 
   none of the elements of $\mathcal{C}$ has a factor in $K[N(\mathcal C)]$.
\end{defn}
 
 \subsection{Parameter identifiability for ODE models}

 Consider an ODE system of the form
 \begin{equation}\label{eq:ODEmodel}
   \Sigma = \begin{cases}
   \bar x' = \bar{f}(\bar x, \bar\mu,\bar u),\\
    \bar y = \bar{g}(\bar x, \bar\mu, \bar u),
    \end{cases}
   \end{equation}
where $\bar x$: a vector of state variables, $\bar y$: a vector of output variables,  $\bar\mu$: a vector of time-independent parameters, $\bar u$: a vector of input variables, and $\bar{f}$ and $\bar{g}$: tuples of elements of $\C(\bar{x}, \bar{\mu}, \bar{u})$.

Bringing $\bar{f}$ and $\bar{g}$ to the common denominator,  write $\bar{f} = \bar{F}/Q$ and $\bar{g} = \bar{G}/Q$, for $F_1, \ldots, F_n, G_1, \ldots, G_m, Q \in \C[\bar{x}, \bar{\mu}, \bar{u}]$.
Consider the (prime, see \cite[Lemma~3.2]{HOPY2018}) differential ideal
\[
 I_\Sigma := [Qx_i' - F_i, Qy_j - G_j,\, 1\leqslant i\leqslant n,\,1\leqslant j \leqslant m] \colon Q^\infty.
\]
Note that every solution of~\eqref{eq:ODEmodel} is a zero of every element of $I_\Sigma$.
 \begin{defn}[Generic solution, cf.~\cite{GL1990,G1992}] 
The image of $(\bar x,\bar y,\bar u)$ under the canonical homomorphism
\[
\C(\bar\mu)\{\bar x,\bar y,\bar u\} \to \C(\bar\mu)\{\bar x,\bar y,\bar u\}/I_\Sigma\]
is called the {\em generic solution} of~\eqref{eq:ODEmodel}.
\end{defn}
 
 Rigorously written definitions of identifiability in analytic terms can be found in \cite[Definition~1]{OPT19} and \cite[Definition~2.5]{HOPY2018}.  \cite[Proposition~3.4]{HOPY2018} implies that the following is an equivalent definition of identifiability, which we will use.
\begin{defn}[Identifiability] A function $h \in\C(\bar\mu)$ is said to be {\em (single-experiment, or SE-) identifiable } for~\eqref{eq:ODEmodel} if, 
for every generic solution $(\bar x^\ast,\bar y^\ast,\bar u^\ast)$ of~\eqref{eq:ODEmodel}, we have $h \in\C\langle\bar y^\ast,\bar u^\ast\rangle$.
\end{defn}

\begin{defn}[Input-output equations,  cf. {\cite[Definition~4.1]{ME2019}}]\label{def:ioeq} 
For a fixed
differential ranking $>$ on $(\bar y,\bar u)$, the set of {\em input-output equations  (IO-equations)} of the system $\Sigma$ from~\eqref{eq:ODEmodel} is the characteristic presentation of 
$I_\Sigma \cap \mathbb{C}\{\bar{y}, \bar{u}\}$.

It can be computed by computing the characteristic presentation $\mathcal{C}$ of $I_\Sigma$ with respect to the differential ranking that is compatible with $>$ and in which any derivative from $\bar x$ is greater than any derivative from $(\bar y,\bar u)$, and
returning $\mathcal{C}\cap \mathbb{C}\{\bar{y}, \bar{u}\}$ (e.g., by the Rosenfeld-Gr\"obner algorithm~\cite{BLOP09}).
\end{defn}

%%%%%%%%%%%%%%%%%%%%%%%%%%%%%%%%%%%%%%%%%%%%%%%%%%%%%%%%%%%%%%%%%%%%

In the Sections~\ref{sec:ce} and~\ref{sec:me}, we will present our two main results, Theorem~\ref{prop:charset} and Theorem~\ref{thm:me}. The former is the main theoretical ingredient for our Algorithm~\ref{alg:main} to find all single-experiment identifiable functions of parameters. The latter is a key to calculating a bound for a sufficient number of experiments to check identifiability of multi-experiment identifiable functions of parameters. Both main results are used in our software implementations referenced in the introduction. We prove our main results in \ref{sec:pfs}.
\section{Main result: Single-experiment identifiability}\label{sec:ce}
In this section, we give an algorithm to compute all functions of the parameters that are identifiable from a single experiment for system~\eqref{eq:ODEmodel}.  
We begin with a construction  in Section~\ref{sec:prep}, which is a refinement of considering Wronskians of 
 monomials
(cf. \cite{G1992, DVJBNP01, MXPW11, OPT19}). 
Using this, we give an algebraic characterization,  Theorem~\ref{prop:charset} (our first main result), of the identifiable functions, which we
turn into Algorithm~\ref{alg:main}.  The proof of Theorem~\ref{prop:charset} can be found in~\ref{sec:pfmain1}.

\subsection{Preparation for Theorem~\ref{prop:charset}}\label{sec:prep}
To  find
the identifiable functions, 
we will begin with a new construction. Let $K$ be a differential field and $k$ a constant subfield such that $\mathbb{C} \subset k$ and let $\bar a = (a_1, \ldots, a_n) \in K^n$.
For  $p \in k\{\bar z\}$, where $\bar z = (z_1, \ldots, z_n)$, such that $p(\bar a) = 0$, we construct a subfield $F(p) \subset K$ as follows:
\begin{enumerate}[itemsep=-0.01in,topsep=0.01in]
    \item Let $W_p$ denote the Wronskian (see Notation~\ref{not:diffrings}\ref{def:Wr}) of the monomials of $p$ evaluated at $\bar a$.
    \item Define $F(p)$ to be the field generated over 
   $\mathbb{C}$
    by (the nonleading) entries in the reduced row echelon form of $W_p$.
\end{enumerate}

\begin{exmp} 
Let $K=\mathbb{C}(x)$, $n=2$, $\bar{a}=(x+1,2/(x+1))$, and $p = z_1' - 2z_1z_2 + 3$. 
Then the  monomials of $p$  are $z_1'$, $z_1z_2$, and $1$, their Wronskian and its evaluation at $\bar{a}$  are
\[\begin{pmatrix}
z_1'&z_1z_2&1\\
z_1''&(z_1z_2)'&0\\
z_1'''&(z_1z_2)''&0
\end{pmatrix}
\quad\text{and}\quad
W_p = \begin{pmatrix}
1&2&1\\
0&0&0\\
0&0&0
\end{pmatrix},
\]
which is already in reduced row echelon form, and so $F(p) = \mathbb{C}(1,2,1) = \mathbb{C}$.
For examples in which $F(p)$ is strictly greater than $\mathbb{C}$, see Example~\ref{ex:illustrate} and Section~\ref{subsec:sf}. There, evaluation of the Wronskian at a point is by differential ideal calculations.
\end{exmp}

For a tuple $\bar p \subset k\{\bar z\}$ of differential polynomials,  
\[F({\bar p}):= \mathbb{C}(F({p})\mid p \in \bar p).\]

\begin{lem}\label{lem:inQa}
  For every $p \in k\{\bar z\}$ such that $p(\bar a) = 0$, we have $F(p) \subset \mathbb{C}\langle \bar a \rangle$.
\end{lem}

\begin{proof}
    Follows from 
    all  entries of $W_p$ being from $\mathbb{C}\langle \bar a \rangle$.
\end{proof}

%%%%%%%%%

\subsection{Statement of main result}\label{sec:main1}

We will now show that the problem of finding the field of identifiable functions is reduced to computing the intersection of fields of constants defined by their generators. This is a key step in  
Algorithm~\ref{alg:main} to find the field of all identifiable functions.
\begin{thm}[Single-experiment identifiability]\label{prop:charset}
  For system~\eqref{eq:ODEmodel}, the field of identifiable functions is equal to 
  \[
    \mathbb{C}(\bar\mu)\cap F(\bar p),
  \] 
  where $\bar p$ is a set of input-output equations of~\eqref{eq:ODEmodel} (Definition~\ref{def:ioeq}).
\end{thm}

%%%%%%%%%%%%%%

\begin{rem} Similarly to Theorem~\ref{thm:me}, the statement of Theorem~\ref{prop:charset} remains true if, in the calculation of $F(\bar p)$, for each $p$, one replaces the Wronskian of the monomials evaluated at $\bar a$ by the Wronskian of any $q_1,\ldots,q_n \in\mathbb{C}\{z\}$ evaluated at $\bar a$ such that $p=\sum_{i=1}^n c_iq_i$ for some $c_1,\ldots,c_n \in k$.
\end{rem}

\subsection{An algorithm for computing all identifiable functions}\label{sec:alg}
In this section, we present an algorithm that computes generators of the field of all identifiable functions of system~\eqref{eq:ODEmodel}. We also give an example following the algorithm step by step.
\begin{algorithm}[H]
\caption{Computing all identifiable functions}\label{alg:main}
\begin{description}
\item[Input] System $\Sigma$ as in~\eqref{eq:ODEmodel}
\item[Output] Generators of the field of identifiable functions of $\Sigma$
\end{description}

\begin{enumerate}[label = \textbf{(Step~\arabic*)}, leftmargin=7.5mm, align=left,itemsep=0in]
    \item\label{step:charset} Compute a set $\bar{p}$ of input-output equations of $\Sigma$ (see Definition~\ref{def:ioeq}).
    \item\label{step:wronskian} For each $p\in\bar p$, compute $\widetilde{W}_p$ the Wronskian of the monomials of $p$.
    Compute $W_p$ by replacing each $y_i^{(j)}$ in $\widetilde{W}_p$ with the $j$-th Lie derivative of $g_i$ with respect to $\Sigma$ ($g_i$'s are the same as in~\eqref{eq:ODEmodel}). 
    \item\label{step:row_echelon} For each $p\in\bar p$, calculate the reduced row echelon form of the matrix $W_p$ and let $F(\bar p)$ be the field generated over $\C$ by all non-leading coefficients of all matrices $W_p$.
    By \cite[Lemma~3.1]{HOPY2018} and Remark~\ref{rem:wronsk_const}, the generators of $F(\bar p)$ belong to $\mathbb{C}(\bar\mu,\bar x)$.
    \item\label{step:intersection} Apply
    Algorithm~\ref{alg:field_intersection} to find generators of $\mathbb{C}(\bar\mu)\cap F(\bar p)$.
    Return these generators.
\end{enumerate}
\end{algorithm}

\begin{rem}
In practice, the runtime of the algorithm depends on the chosen ranking, and it would be interesting to have a way to choose the ranking based on the problem.
\end{rem}

\begin{exmp}[Computing identifiable functions -- illustration]\label{ex:illustrate}
To illustrate, we will follow  Algorithm~\ref{alg:main} for the system:
\[
\Sigma=\begin{cases}
x_1'=0,\\
x_2' = x_1x_2 + ax_1u + bu,\\
y = x_2
\end{cases}
\]
where $\bar x = (x_1,x_2)$, $\bar y=(y)$,   $\bar\mu=(a,b)$, $\bar{u}=(u)$.
This system is a variant of the example from~\cite[Section~5]{second_paper}.
\begin{enumerate}[label = \textbf{(Step~\arabic*)}, leftmargin=1.5mm, align=left,itemsep=0in]
    \item
For the elimination differential ranking with $x_1>x_2>y>u$,  a calculation shows that
\[x_1y - ax_1u - y' - bu,\quad x_2-y,\quad yy'' - auy'' - y'^2 + au'y' - buy' + bu'y
\]
is a monic characteristic presentation for $I_\Sigma$.
Therefore, $\bar{p}=(p)$, where $p =yy'' - auy'' - y'^2 + au'y' - buy' + bu'y$. 
\item The Wronskian $\widetilde{W}_p = \Wrn(u'y, uy', u'y', y'^2, uy'', yy'')$ is computed 
(too large to be displayed here).
Then,  to compute $W_p$, all  derivatives of $y$ are replaced with the corresponding Lie derivatives of $x_2$, for example:
\begin{align*}
  y &\to x_2,\; y' \to x_1x_2 + ax_1u + bu,\\
  y'' &\to x_1(x_1x_2 + ax_1u + bu) + ax_1u' + bu'.
\end{align*}
\item  A calculation shows that the corresponding reduced row echelon form is:
\[
\begin{pmatrix}
1&0&0&0&-x_1&-ax_1-b\\
0&1&0&0&x_1&ax_1+b\\
0&0&1&0&1&0\\
0&0&0&1&0&1\\
0&0&0&0&0&0\\
0&0&0&0&0&0
\end{pmatrix}.
\]
Therefore, $F(\bar{p})=\mathbb{C}(ax_1+b,x_1)$. 
\item By Theorem~\ref{prop:charset}, the field of identifiable functions is
\[
  \mathbb{C}(a, b)\cap \mathbb{C}(ax_1 + b, x_1).
\]
Applying Algorithm~\ref{alg:field_intersection}, we find that
\begin{equation}\label{eq:Caxb}
  \mathbb{C}(a, b)\cap \mathbb{C}(ax_1 + b, x_1) = \mathbb{C},
\end{equation}
so there are no nontrival identifiable functions in this model.
\end{enumerate}
\end{exmp}

%%%%%%%%%%%%%%%%%%%%%%%%%%%%%%%%%%%

\section{Main result: Multi-experiment identifiability}\label{sec:me}

In this section, we show that the coefficients input-output equations generate the field of multi-experiment identifiable function and
derive a generally {\em tight} upper bound for the number of independent experiments for system~\eqref{eq:ODEmodel} sufficient to recover the field of multi-experiment
identifiable functions of parameters.
These results are
stated in Section~\ref{sec:memain}  and proven in~\ref{sec:pfmain2}. We apply 
them
to specific examples  from the literature in Section~\ref{sec:examples}.
The tightness of the bound from the mathematical point of view is discussed in the appendix.

\subsection{Preparation for Theorem~\ref{thm:me}}\label{sec:memain} 

\begin{defn}[Input-output identifiable functions]  A function of parameters $h\in\mathbb{C}(\bar\mu)$ in system~\eqref{eq:ODEmodel} is said to be {\em input-output (IO) identifiable} if $h$ can be expressed as a rational function of the coefficients of the IO-equations of system~\eqref{eq:ODEmodel} (see Definition~\ref{def:ioeq}), see also  \cite[Definition~2]{OPT19} and
 \cite[Corollary~1]{ident-compare}. 
\end{defn}
As shown in \cite[Section~4.1]{ident-compare}, every identifiable function is input-output identifiable but not every input-output identifiable function is necessarily identifiable.

\begin{defn}[Multi-experiment identifiability, cf.~\cite{GenSSI2}]  
 A function of parameters $h\in\mathbb{C}(\bar\mu)$ in system~\eqref{eq:ODEmodel} is said to be {\em multi-experiment identifiable (ME-identifiable)}  if there exists $N \geqslant 1$ such that $h$ is identifiable in the following ``$N$-experiment'' system
\begin{equation}\label{eq:ME}
\Sigma_N:=\begin{cases}
\bar x_i'=f(\bar x_i,\bar\mu,\bar u_i),\\
y_i=g(\bar x_i,\bar\mu,\bar u_i),
\end{cases}
\quad 1\leqslant i\leqslant N.
\end{equation}
We also say that $h$ is $N$-experiment identifiable in this case.
\end{defn}
\begin{exmp}[Illustrating the definition] Consider the  system (intentionally simple to illustrate the definition)
\begin{equation}\label{eq:shartbound}
\begin{cases}
 x_1'=0\\
y_1=x_1\\
y_2=\theta x_1+\theta^2,
\end{cases}
\end{equation}
in which $\theta$ is the unknown parameter.
By \cite[Example~2.14]{HOPY2018}, $\theta$ is not (globally) identifiable. Consider now the corresponding $2$-experiment system
\[\begin{cases}
x_{1,1}' = x_{2, 1}' =0\\
y_{1,1}=x_{1,1}\\
y_{1,2}=\theta x_{1,1}+\theta^2\\
y_{2,1}=x_{2,1}\\
y_{2,2}=\theta x_{2,1}+\theta^2;
\end{cases}\]
 now $\theta =\dfrac{y_{2, 2} - y_{1, 2}}{y_{2, 1} - y_{1, 1}}$ and so is identifiable.
\end{exmp}

\begin{rem}\label{rem:SIAN}
SIAN~\cite{SIAN} 
(see also~\cite[Remark~2.3]{ident-compare}) 
is software that can check (SE-) global and local identifiability 
of any given function $h \in \C(\bar{\mu})$ of parameters
of an ODE model $\Sigma$.
If $h$ is globally ME-identifiable, then, running SIAN on models of the form $\Sigma_N$ (see~\eqref{eq:ME}) for $N = 1, 2, \ldots$, one will in principle eventually find this out.
However, if $h$ is not globally ME-identifiable,  one will not be able to conclude this from assessing SE-identifiability  of $\Sigma_N$ without a bound on the number of experiments (provided by Theorem~\ref{thm:me}).

On the other hand, one could use SIAN to find the sufficient number of experiments \textbf{given} a set of generators of the field of ME-identifiable functions.
Indeed, for each of these generators, there is an $N$ such that the generator is SE-identifiable in $\Sigma_N$, so the sufficient number of experiments can be taken as the maximum of these $N$s.
However, this approach works only if  generators of the field of ME-identifiable functions  are known in advance.
Theorem~\ref{thm:me} and an algorithm to compute IO-equations  (Definition~\ref{def:ioeq}) yield an algorithm to find such generators.
\end{rem}

\subsection{Statement of main result}
\begin{thm}[Multi-experiment identifiability]\label{thm:me} A function of parameters $h \in\mathbb{C}(\bar\mu)$ in system~\eqref{eq:ODEmodel} is multi-experiment identifiable if and only if it is input-output identifiable in system~\eqref{eq:ODEmodel}.

Moreover, if $h$ is multi-experiment identifiable, then, for all 
\[
  N\geqslant\max\limits_{1\leqslant i \leqslant m}(s_i-r_i+1),
\] 
$h$ is identifiable in the $N$-experiment system,  where 
 $s_i$ and $r_i$ are defined by the following: \begin{itemize}[itemsep=0in,topsep=0in]
\item $\bar p=p_1,\ldots,p_m$ is a set of input-output equations of system~\eqref{eq:ODEmodel}, and for all $i$, $1\leqslant i\leqslant m$,
\begin{itemize}[itemsep=0in,topsep=-0.05in]
\item we write 
\begin{equation}\label{eq:p_decomposition}
  p_i = f_{i, s_i + 1} + \sum\limits_{j=1}^{s_i} c_{i, j} f_{i, j},
\end{equation}
where $f_{i,j} \in\mathbb{C}\{\bar{y},\bar{u}\}$  and linearly independent over $\mathbb{C}$ (so, $s_i$ is the length of such a presentation of $p_i$ minus $1$), 
\item  $r_i:=
\rank\Wrn(f_{i,1}(\bar{y},\bar{u}),\ldots,f_{i,s_i}(\bar{y},\bar{u}))$ modulo $I_\Sigma$.
\end{itemize}
\end{itemize}
\end{thm}

\begin{exmp}[Degenerate Wronskian]\label{ex:two}
The goal of this intentionally  simple 
example is to demonstrate that the Wronskians in the theorem can indeed be singular.
Consider  system~\eqref{eq:shartbound} again.
A calculation shows that
\[\bar{p}=\left\{y_1',y_2-\theta y_1-\theta^2\right\}\]
is a set of IO-equations for~\eqref{eq:shartbound}.
Then $m=2$, $s_1$=0, and $s_2=2$. We have 
\[
  \Wrn(y_1,1)=
  \begin{pmatrix}
  y_1&1\\
  y_1'& 0
  \end{pmatrix}\mod I_\Sigma=\begin{pmatrix}
  x_1&1\\
  0& 0
  \end{pmatrix},
\] 
and so $r_2 = 1$.  From \cite[Example~2.14]{HOPY2018}, $\theta$ is not (globally) identifiable (so, we cannot take $N=1$). By Theorem~\ref{thm:me}, for all \[N \geqslant 2-1+1 =2,\]
the field of ME-identifiable functions
$\mathbb{C}(\theta,\theta^2)=\mathbb{C}(\theta)$
is $N$-experiment identifiable.
\end{exmp}

\begin{rem}
  In some works (e.g., \cite[Section~3.1]{Bearup2013}), it was suggested that the Wronskians of monomials in a characteristic set be always of corank one ($r_i = s_i$ in the notation of Theorem~\ref{thm:me}).
  As Example~\ref{ex:two} (see also Sections~\ref{subsec:sf} and \ref{sec:LV-mult}) shows, this is not the case.
\end{rem}

\subsection{Computational aspects}

\begin{rem}[Dependence on decomposition~\eqref{eq:p_decomposition}]
  For fixed input-output equations $p_1, \ldots, p_m$, the bound given by Theorem~\ref{thm:me} may depend on the choice of decomposition~\eqref{eq:p_decomposition}.
  In~\ref{app:rep5}, we give an algorithm to compute a representation yielding the best possible bound (compared to other representations).
  We use this algorithm in our implementation.
\end{rem}

\begin{rem}[Computing the bound]\label{rem:alg_bound} 
The rank of the Wronskian matrix from
Theorem~\ref{thm:me} can be found by:
\begin{enumerate}[itemsep=-0.04in,topsep=0.02in]
\item Calculating the Wronskian matrix in $\bar y,\bar u$, 
\item For each matrix entry, computing its differential remainder \cite[Section~I.9]{Kol}
 with respect to the characteristic set defined by $\Sigma$, and 
\item  Applying a (symbolic) algorithm for rank computation. 
\end{enumerate}
The correctness follows from~\cite[Lemma~3.1]{HOPY2018}.
Before computing the rank, one can evaluate the Wronskian at a point. 
Since the rank cannot increase after an evaluation, the resulting bound will always be correct although might be larger than the bound from~Theorem~\ref{thm:me}.
\end{rem}

\begin{rem}\label{rem:const}
The bound for $N$ from Theorem~\ref{thm:me} can be improved if some of the output variables are constant as discussed in Section~\ref{sec:LV-mult}.
Constant outputs arise, e.g., to encode the case of constant inputs, which is common in some application domains~\cite{VECB}.
The general idea of the refinement is first to treat the constant outputs as parameters, apply Theorem~\ref{thm:me} to the rest of the outputs, and then use simultaneous rational interpolation to extract the coefficients with respect to the constant outputs.
\end{rem}

\section{Examples}\label{sec:examples}

We illustrate our results with 3 examples.
In  Section~\ref{subsec:two_comp}, Lotka-Volterra model with control, we show that the SE-identifiable and ME-identifiable functions coincide, so one can find the generators of the field of identifiable functions from the coefficients of the IO-equations.
The second example (Section~\ref{subsec:sf}) is a chemical reaction exhibiting the slow-fast ambiguity~\cite{VR1988}. 
Here, the bound from Theorem~\ref{thm:me} is exact, and yields that all parameters are identifiable from 2 experiments.
 In Section~\ref{sec:LV-mult}, we show another Lotka-Volterra model, for which some of the parameters become identifiable only after 2 experiments.
For other models with more ME-identifiable functions than SE-identifiable ones, we refer to~\cite[Section~III]{VECB}.
Finally, in Section~\ref{sec:SEIR}, we apply our results to the SEIR epidemiological model studied in~\cite{TL18}.

All the computations for the examples in this section can be performed automatically using our implementation. The corresponding files can be found in the \texttt{examples} folder in the repository~\url{https://github.com/pogudingleb/AllIdentifiableFunctions}.

\subsection{Lotka-Volterra model with control}\label{subsec:two_comp}
Consider the following system
\[
\Sigma= \begin{cases}
x_1'= a x_1 - bx_1x_2,\\ 
x_2' = -cx_2 + dx_1x_2+eu,\\
y=x_1,
\end{cases}
\]
in which $a,b,c,d,e$ are the unknown parameters and $u$ is the input (control).
With Theorem~\ref{thm:me}, we  show that, for this model, the fields of SE-identifiable and of ME-identifiable functions coincide.
A computation shows that the IO-equation is:
\[
\bar{p} = (yy'' - y'^2 - dy^2y' + cyy' + ady^3 - beuy^2 - acy^2),
\]
so, in the notation of Theorem~\ref{thm:me}, $m=1$ and, for $f_1 = y^2y'$, $f_2=yy'$, $f_3 = y^3$, $f_4 = uy^2$, $f_5 = y^2$, and $f_6 = yy''-y'^2$, we have $s_1 = 5$. 
A computation shows that 
\[
 r_1 := \rank(\Wrn(f_1,f_2,f_3,f_4,f_5) \mod I_\Sigma)=5.
\]
By Theorem~\ref{thm:me}, for any \[
N \geqslant 5-5+1 =1
\] 
the ME-identifiable functions are identifiable from $N$ experiments (cf. \cite[Main Results~1 and~2]{OPT19}).
In particular, 1 experiment is sufficient.
Hence, by Theorem~\ref{thm:me}, the field of  SE-identifiable functions is $
\mathbb{C}(d, c, ad, be, ac)= \mathbb{C}(a, be, c, d)$.

%%%%%%%%%%%%%%

\subsection{Slow-fast ambiguity in chemical reactions}\label{subsec:sf}
In this example, we consider the system \cite[Section~A.1,equation~(3)]{SIAN}.
This system originates from the following chemical reaction network 
\cite[equation~(1.1)]{VR1988}:
\[
A \xrightarrow{k_1} B \xrightarrow{k_2} C.
\]
Then the quantities $x_A, x_B$, and $x_C$ of species satisfy the  system:
\begin{equation}\label{eq:chem}
\begin{cases}
  x_A' = -k_1 x_A,\\
  x_B' = k_1 x_A - k_2 x_B,\\
  x_C' = k_2 x_B.
\end{cases}
\end{equation}
The observed quantities will be
\begin{itemize}[itemsep=-0.07pt,topsep=0.005in]
    \item $y_1 = x_C$, the concentration of $C$;
    \item $y_2 = \varepsilon_A x_A + \varepsilon_B x_B + \varepsilon_C x_C$, which may represent a property of the mixture, e.g. absorbance or conductivity~\cite[p.~701]{VR1988}.
\end{itemize}
As explained in~\cite[p.~701]{VR1988}, in practice, $x_B$ might be hard to isolate, so $\varepsilon_B$ is also an unknown parameter, while the values $\varepsilon_A$ and $\varepsilon_C$ can be assumed to be known but could depend on $A$, $C$, and the details of the experimental setup.
The assumption that $\varepsilon_A$ and $\varepsilon_C$ are known can be encoded into the ODE system by making them state variables with zero derivatives and adding outputs to make them observable.
This will yield the following final ODE model (the same as~\cite[Section~A.1,equation~(3)]{SIAN}):
\begin{equation}\label{eq:chem_main}
\Sigma =\begin{cases}
  x_A' = -k_1 x_A,\\
  x_B' = k_1 x_A - k_2 x_B,\\
  x_C' = k_2 x_B,\\
  \varepsilon_A' = \varepsilon_C'=0,\\
  y_1 = x_C,\\
  y_2 = \varepsilon_A x_A + \varepsilon_B x_B + \varepsilon_C x_C,\\
  y_3 = \varepsilon_A,\\
  y_4 = \varepsilon_C,
\end{cases}
\end{equation}
where $\bar{x} = (x_A, x_B, x_C, \varepsilon_A, \varepsilon_C)$, $\bar{y} = (y_1, y_2, y_3, y_4)$, and $\bar{\mu} = (k_1, k_2, \varepsilon_B)$.
As noted in~\cite{VR1988} (see also~\cite[Section~A.1]{SIAN}), this model has slow-fast ambiguity: it is possible to recover a pair of numbers $\{k_1, k_2\}$ from the observations but impossible to know which one is $k_1$ and which one is $k_2$.
A similar phenomenon occurs in epidemiological models, see~\cite[Proposition~2]{TL18}.

We start with assessing the \textbf{SE-identifiability} of the model~\eqref{eq:chem_main} using Algorithm~\ref{alg:main} to find the field of identifiable functions.
For~\ref{step:charset}, a calculation in {\sc Maple} shows that the following set $\bar p =\{p_1,p_2,p_3,p_4\}$ is a set  of IO-equations of~\eqref{eq:chem_main}:
\begin{gather*}
p_1=k_1k_2(y_2-y_1y_4) - \varepsilon_{B}k_1y_1' - k_2y_1'y_3 - y_1''y_3,\\
p_2= y_1''' + (k_1+k_2)y_1'' +  k_1k_2y_1',\quad p_3=y_3',\quad
p_4 =y_4'.
\end{gather*}
In~\ref{step:wronskian} and~\ref{step:row_echelon}, we compute the reduced row echelon forms of $W_{p_1} = \Wrn(y_2, y_1y_4, y_1', y_1'y_3, y_1''y_3)$ and 
$W_{p_2} = \Wrn(y_1'', y_1', y_1''')$ modulo the equations 
$\Sigma$ and obtain the matrices
\begin{gather*}
\begin{pmatrix}
1 & 0 & 0 & 0 & k_1k_2\\
0 & 1 & 0 & 0 & -k_1k_2\\
0 & 0 & 1 & \varepsilon_A & -(\varepsilon_Ak_2 + \varepsilon_Bk_1)\\
0 & 0 & 0 & 0 & 0\\
0 & 0 & 0 & 0 & 0
\end{pmatrix},\ \ \ 
\begin{pmatrix}
1 & 0 & -k_1k_2\\
0 & 1 & -(k_1+k_2)\\
0 & 0 & 0
\end{pmatrix},
\end{gather*}
respectively.
$W_{p_3}$ and $W_{p_4}$ are $1\times 1$ matrices with the reduces row echclon form $(1)$.
Therefore, 
\[
  F(\bar{p}) = \C(k_1 + k_2, k_1k_2, \varepsilon_A, \varepsilon_Ak_2 + \varepsilon_B k_1).
\]
Before going to~\ref{step:intersection}, we 
show that this intermediate result of computation can  provide additional insights, for example, recover the parameter transformation corresponding to the slow-fast ambiguity~\cite[equation~(1.3)]{VR1988}.
From the proof of Theorem~\ref{prop:charset},  $F(\bar{p})$ consists of identifiable constants.
So, any parameter transformation  induces an automorphism $\alpha$ of the constants over $F(\bar p)$.
Since $k_1 + k_2$ and $k_1k_2$ are identifiable, $\alpha(k_1) = k_1$ and $\alpha(k_2) = k_2$ or $\alpha(k_1) = k_2$ and $\alpha(k_2) = k_1$.
Consider the latter case.
Since $\varepsilon_A \in F(\bar{p})$, we have $\alpha(\varepsilon_A) = \varepsilon_A$.
Hence,
\[
\varepsilon_Ak_2 + \varepsilon_Bk_1 = \alpha(\varepsilon_Ak_2 + \varepsilon_Bk_1) = \varepsilon_Ak_1 + k_2 \alpha(\varepsilon_B),
\]
so  $\alpha(\varepsilon_B) = \varepsilon_A  + \frac{k_1(\varepsilon_B - \varepsilon_A)}{k_2}$, giving the transformation~\cite[(1.3)]{VR1988}:
\begin{equation}\label{eq:transform}
k_1 \to k_2,\  k_2 \to k_1,\  \varepsilon_A \to \varepsilon_A,\ \varepsilon_B \to \varepsilon_A  + \frac{k_1(\varepsilon_B - \varepsilon_A)}{k_2}.
\end{equation}
Finally, in~\ref{step:intersection}, we compute
\[
  \C(k_1, k_2, \varepsilon_B) \cap F(\bar{p}) = \C(k_1k_2, k_1 + k_2).
\]

Now we will consider 
model~\eqref{eq:chem_main} in the \textbf{multi-experiment setup} in which one is allowed to perform several experiments with the same $k_1, k_2, \varepsilon_B$ but different initial concentrations and $\varepsilon_A, \varepsilon_C$.
We will show that, in this setup, the ambiguity can be resolved by one extra experiment.
The first part of Theorem~\ref{thm:me} implies that the field of ME-identifiable functions is generated by the coefficients of $\bar{p}$, so it equals
\[
\C(k_1k_2, \varepsilon_Bk_1, k_2, k_1 + k_2) = \C(k_1, k_2, \varepsilon_B).
\]
Therefore, \emph{all} the parameters can be identified from several experiments.
Now we  use the bound from Theorem~\ref{thm:me} to find the number of experiments sufficient to make all the parameter identifiable.
In the notation of the theorem, for $i=1$, we  take 
\[
f_{1,1} = y_2-y_1y_4,\ \  f_{1,2} = y_1',\ \  f_{1,3}=y_1'y_3,\ \  f_{1,4}=y_1''y_3,
\]
and so $s_1=3$.
A calculation in {\sc Maple} shows that
\[
  r_1 = \rank\Wrn(f_{1,1},f_{1,2},f_{1,3})\mod I_\Sigma = 2.
\]
so the Wronskian does not always have full rank in practical examples either. 
For $i=2$,
  $f_{2,1} = y_1''$, $f_{2,2}=y_1'$,
so $s_2=2$, and 
\[
  r_2 = \rank\Wrn(f_{2,1},f_{2,2}) \mod I_{\Sigma} = 2.
\]
Finally, $f_{3,1} = y_3'$ and $f_{4,1} = y_4'$, and so $s_3 = s_4 = 0$. 
Thus, all parameters are $N$-identifiable for all 
\[
  N\geqslant \max(3-2+1,2-2+1,0-0+1,0-0+1)=2.
\]
This bound  
is tight because, as we demonstrated earlier, neither of the parameter is identifiable from a single experiment.

%%%%%%%%%%%%%%%%%%%%%%%%%%%%%%%%%%%%%%%%%%%%%%%%%%%%%%

\subsection{Lotka-Volterra model with ``mixed'' output}\label{sec:LV-mult}

In this example, we will illustrate the refinement of the bound on the number of experiments mentioned in Remark~\ref{rem:const} on the following variant of the Lotka-Volterra model:
\[
\Sigma= \begin{cases}
x_1'=ax_1-x_1bx_2,\\ x_2' =-cx_2+dx_1x_2,\\
y=ex_1+fx_2,
\end{cases}
\]
where we assume that $a,b,c,d,e$ are unknown parameters and $f$ is a known parameter that takes different values if multiple experiments are conducted. 
In the context of our differential algebra setup, this can be encoded as follows:
\begin{equation}\label{eq:LVmix}
\Sigma= \begin{cases}
x_1'=ax_1-bx_1x_2,\\ x_2' =-cx_2+dx_1x_2,\\
f' = 0,\\
y_1=ex_1+fx_2,\\
y_2 =f
\end{cases}
\end{equation}
Our implementation shows that the field of ME-identifiable functions is
\begin{equation}\label{eq:LVmix_gens}
  \mathbb{C}(a,\; b,\; c,\; d/e).
\end{equation}
In particular, $a, b, c$ are ME-identifiable, $d$ and $e$ are not but their ratio is.

We now discuss the number of experiments for globally identifying the functions~\eqref{eq:LVmix_gens}.
A straightforward application of Theorem~\ref{thm:me} yields a bound $35$ (Wronskian of dimension $51$ and rank $17$). 
35 could be viewed as a rather high number of experiments and is far from the actual number ($2$, as shown below).

We can get a better bound equal to $4$ using the same Theorem~\ref{thm:me} as follows.
Observe that, since $y_2$ is constant, then there will be the following input-output equations for the model: $y_2' = 0$ and $p = 0$, where $p$ is a differential polynomial $y_1$ and $y_2$ over $\mathbb{C}(\bar{\mu})$ of zero order in $y_2$.
We observe that, if one replaces $y_2$ in $p = 0$ with $f$, the resulting equations will be the input-output equations for the following simplified model, in which $f$ is considered as a scalar parameter:
\begin{equation}\label{eq:LVmix_simpl}
\Sigma= \begin{cases}
x_1'=ax_1-bx_1x_2,\\ 
x_2' =-cx_2+dx_1x_2,\\
y=ex_1+fx_2,\\
\end{cases}
\end{equation}
Our implementation shows that the bound for this model is one, so SE-identifiable and ME-identifiable functions for this model are the same.
In particular the coefficients of the monic input-output equation of~\eqref{eq:LVmix_simpl} are identifiable from a single experiment.
These coefficients are rational function in $f$ over $\mathbb{C}(a, b, c, d, e)$.
We write them as $C_1/C, \ldots ,C_s/ C$, where $C, C_1, \ldots, C_s$ are polynomials in $f$ over  $\mathbb{C}(a, b, c, d, e)$, and $C$ is monic.
We denote the number coefficients not belonging to $\mathbb{C}$ in $C,C_1,\ldots,C_s$ by $n, n_1, \ldots, n_s$, respectively.
Then these coefficients can be determined uniquely from 
\[
  \max\left( n + \min\limits_{1\leqslant i\leqslant s} n_i, \max\limits_{1\leqslant i\leqslant s} n_i \right).
\]
evaluations for different values of $f$.
To show this, assume that $n_1 = \min\limits_{1\leqslant i\leqslant s} n_i$.
Then $n + n_1$ evaluations are sufficient to reconstruct coefficients of $C_1/C$ as a rational function in $f$.
Then, once the coefficients of $C$ are known, evaluations of $C_2 /C, \ldots, C_s / C$ can be used to find the coefficients of $C_2, \ldots, C_s$ via polynomial interpolation.
In this example, $n = 2$, $\min\limits_{1\leqslant i\leqslant s} n_i = 2$, and $\max\limits_{1\leqslant i\leqslant s} n_i = 4$, so four evaluations (that is, four experiments with different known values of $f$ will be enough).

The obtained bound $4$ is  close to the exact bound $2$, which can be obtained using Theorem~\ref{thm:me} and SIAN as follows.
Using SIAN, we obtain that $a, b, c, d/e$ are only locally identifiable (from one experiment), so $N > 1$. 
Running SIAN for 2 experiments shows that the functions~\eqref{eq:LVmix_gens} are 2-experiment globally identifiable.
Since from Theorem~\ref{thm:me} we know that these functions generate all ME-identifiable functions, we conclude that $N = 2$.
Replication of the system makes it substantially more challenging for SIAN, so this approach might be impractical if $N$ is large, while computing the bound above may be feasible.

%%%%%%%%%%%%%%%%%%%%%%%%%%%

\subsection{SEIR epidemiological model}\label{sec:SEIR}
Structural identifiability of the following epidemiological model has been considered in~\cite[Equation~2.2]{TL18}
\begin{equation}\label{eq:SEIR}
    \begin{cases}
      S' = -\beta\frac{SI}{N},\\
      E' = \beta \frac{SI}{N} - \eta E,\\
      I' = \eta E - \alpha I,\\
      R' = \alpha I,
    \end{cases}
\end{equation}
where $N$ is the total population which is constant and known.
The following  two setups are considered in~\cite{TL18}:
\begin{itemize}
    \item \emph{Prevalence observation.}
    In this case, these is an output $y_1 = I$.
    We also add $y_2 = N$ to account for the fact that $N$ is known.
    Our implementation shows that the bound from Theorem~\ref{thm:me} is equal to one, so the fields of ME-identifiable and SE-identifiable functions coincide.
    It also finds that these fields are equal to $\mathbb{C}(\alpha \eta, \alpha + \eta, \beta, N)$.
    
    \item \emph{Cumulative incidence observation.}
    In this case, the observed quantity is $\int \eta E \operatorname{dt}$.
    This can be encoded by introducing a new state variable $C$ with $C' = \eta E$ and 
    the outputs $y_1 = C$ and $y_2 = N$.
    Our algorithm again shows that the fields of ME-identifiable and SE-identifiable functions coincide and that they equal  $\mathbb{C}(\alpha, \beta, \eta, N)$, so all of the parameters are globally identifiable.
\end{itemize}
These results confirm the findings of~\cite{TL18} obtained from analysis of input-output equations (that is, for ME-identifiability) and show that they are valid for SE-identifiability as well. 

%%%%%%%%%%%%%%
\section*{Acknowledgments}
We are grateful to Julio Banga and Alejandro Villaverde for learning from them about the importance of  multi-experiment parameter identifiability at the AIM workshop ``Identifiability problems in systems biology''. We thank the editors and referees for their useful comments and suggestions.
This work was partially supported by the NSF grants  CCF-1564132, CCF-1563942, DMS-1760448, DMS-1760413, DMS-1853650, DMS-1665035, DMS-1760212, DMS-1853482, and DMS-1800492, and by the Paris Ile-de-France Region.

\bibliographystyle{abbrvnat}
\bibliography{bibdata}

\appendix
\section{Proofs of the main results}\label{sec:pfs}
\subsection{Single-experiment identifiability}\label{sec:pfmain1}
In this section, we will prove our first main result, Theorem~\ref{prop:charset}.
\begin{rem}[$F(p)$ is generated by first integrals]\label{rem:wronsk_const}
  Lemma~\ref{lem:wronsk_const} implies that $F(\bar{p})$ consists of constants (that is, first integrals).
  \cite[Lemma~2]{OPT19} implies that  $F(\bar{p}) \subset \mathbb{C}(\bar{\mu}, \bar{x})$.
\end{rem}

\begin{lem}\label{lem:wronsk_const}
  In the setup of Section~\ref{sec:prep}, $F(p) \subset \Const(K)$.
\end{lem}

\begin{proof}
   We will show that the space of linear relations between columns of $W_p$ is defined over $\Const(K)$. This will imply that all the entries in the reduced row echelon form of $W_p$ are constants.
   Let $X$ be a maximal set of linearly independent columns of $W_p$, and let $Y$ denote the rest of the columns.
   The monomials corresponding to $X$ are linearly independent over $\Const(K)$ because any such dependence would yield a dependence of $X$.
   For each $v \in Y$, there exists a column dependence of $X \cup \{v\}$, unique up to scaling.
   These dependencies span the space of column dependencies of $W_p$.
   \cite[Chapter~II, Section~1, Theorem~1]{Kol} implies that the monomials corresponding to $X \cup \{v\}$ are dependent over $\Const(K)$.
   Therefore, the corresponding column dependence can be chosen to be over $\Const(K)$ as well.
\end{proof}

\begin{proof}[Proof of Theorem~\ref{prop:charset}]
To prove the theorem, 
we will show that, for all differential fields $k\subset K$ with $\mathbb{C} \subset  k \subset \Const(K)$ and $k$ being algebraically closed in $K$, every $n$, and every tuple $\bar a \in K^n$,
\[k\cap F({\bar p})  = k\cap \mathbb{C}\langle \bar a \rangle,\]
where $\bar p := \{p_1, \ldots, p_m\} \subset k\{z_1,\ldots,z_n\}$ is a characteristic set of the prime ideal of all differential polynomials vanishing at $\bar a$. This is then applied to $k=\mathbb{C}(\bar\mu)$ and the differential field $K$ generated over $k$ by the $(\bar y,\bar u)$-components (denoted by $\bar a$) of a generic solution of~\eqref{eq:ODEmodel}.

     Lemma~\ref{lem:inQa} implies that $k\cap F({\bar p})  \subset k\cap \mathbb{C}\langle \bar a \rangle$.
  Assume that \[k\cap \mathbb{C}\langle \bar a \rangle  \supsetneq k\cap F({\bar p}) \] and let $b \in k\cap \mathbb{C}\langle \bar a \rangle  \setminus k\cap F({\bar p}) $.
  
  Recall (see \cite[Section~2]{M1996}) that a differential field $K$ is 
  differentially closed if: for all $m$ and finite $G \subset K\{w_1,\ldots,w_m\}$, if there exists $L\supset K$ such that $G=0$ has a solution in $L$, then $G=0$ has a solution in $K$.
  Let $\Kdiff$ be a differential closure of $K$, that is, a differentially closed field containing $K$ that embeds into any other differentially closed field containing $K$.
  
  We have $\Kdiff\supset \kacl$, the algebraic closure of $k$, and $\kacl \cap K = k$.
  Since $b \not\in F({\bar p})$, there exists an automorphism $\alpha\colon \Const(\Kdiff) \to \Const(\Kdiff)$ such that $\alpha|_{F({\bar p})} = \operatorname{id}$ and $\alpha(b) \neq b$.
  We pick such an $\alpha$ and extend it to a differential automorphism of $\Kdiff$ and denote the extension by $\alpha$ as well.
  
  For a vector $K$-subspace $V$ of $K^n$ with $\mathbb{C} \subset K$, we denote the field of definition of the subspace over $\mathbb{C}$ by $\operatorname{FD}(V)$.
Recall that $V$ has a $K$-basis $e_1, \ldots, e_\ell$ of $V$ such that $e_1, \ldots, e_\ell \in \operatorname{FD}(V)^n$.

  Fix $1 \leqslant i \leqslant m$. Let $V_{p_i}$ denote the right kernel of $W_{p_i}$.
 Note that $V_{p_i}$ is defined over $\Const(K)$. 
  Since $p_i(\bar a) = 0$, the vector of  coefficients of $p_i$ belongs to $V_{p_i}$.  Note that $\operatorname{FD}(V_{p_i})=F(p_i)$.
  By the preceding paragraph, there exist $r_{i, 1}, \ldots, r_{i, N_i} \in \operatorname{FD}(V_{p_i}) \{z_1, \ldots, z_n\}$ such that 
  \begin{itemize}[itemsep=-0.03in,topsep=0in]
      \item for every $1 \leqslant j \leqslant N_i$, the vector of the coefficients of $r_{i, j}$ belongs to $V_{p_i}$ (in particular, $r_{i, j} (\bar a) = 0$);
      \item $p_i$ is a $K$-linear combination of $r_{i, 1}, \ldots, r_{i, N_i}$.
  \end{itemize}
  
  Since $b \in \mathbb{C}\langle \bar a \rangle$, there exist differential polynomials $R_1, R_2 \in \mathbb{C}\{ \bar z\}$ such that $b = \frac{R_1(\bar a)}{R_2(\bar a)}$.
  We write $H = S_1\cdot\ldots\cdot S_m\cdot I_1\cdot \ldots\cdot I_m$, where $I_i$ and $S_i$ are the initial and separant of $p_i$.
  Since $b\in k$, $bR_2-R_1 \in k\{\bar x\}$. Since additionally  $bR_2(\bar a)-R_1(\bar a) =0$,
  \[
    H(b R_2 - R_1) \in \sqrt{[\bar p]}.
  \]
  Since, for every $1 \leqslant i \leqslant \ell$, $p_i \in [r_{i, 1}, \ldots, r_{i, N_i}]$, we have
  \begin{equation}\label{eq:ideal_membership_char}
    H(b R_2 - R_1) \in \sqrt{[r_{1, 1}, \ldots, r_{m, N_m}]}.
  \end{equation}
  We apply $\alpha$ to~\eqref{eq:ideal_membership_char} and use  that $r_{i, j}$'s are invariant under $\alpha$:
  \begin{equation}\label{eq:after_alpha_char}
        \alpha(H)(\alpha(b) R_2 - R_1) \in \sqrt{[r_{1, 1}, \ldots, r_{m, N_m}]}.
  \end{equation}
  We multiply~\eqref{eq:ideal_membership_char} by $\alpha(H)$ and~\eqref{eq:after_alpha_char} by $H$, and subtract.
  We obtain
  \[
    H \alpha(H) R_2 (\alpha(b) - b) \in \sqrt{[r_{1, 1}, \ldots, r_{m, N_m}]}.
  \]
  Every element of $\sqrt{[r_{1, 1}, \ldots, r_{m, N_m}]}$ vanishes at $\bar a$ since every $r_{i, j}$ vanishes at $\bar a$.
  Since $H(\bar a) \neq 0$ and $R_2(\bar a)(\alpha(b) - b) \neq 0$, it is sufficient to show that $\alpha(H)(\bar a) \neq 0$ to arrive at contradiction.

  Assume  there are $1 \leqslant i \leqslant m$ and $h \in \{S_i, I_i\}$ such that $\alpha(h)(\bar a) = 0$.
  Consider the sets $M_0$ and $M_1$ of monomials of $\alpha(h)(\bar a)$ (or, equivalently, of $h(\bar a)$) and $p_i(\bar a)$, respectively.
  Observe that there is a monomial $A$ in $\bar a$ with $AM_0 \subset M_1$ because
  \begin{itemize}[itemsep=-0.05in,topsep=0.05in]
      \item if $h = S_i$, take $A$  to be $\lead p_i(\bar a)$;
      \item if $h = I_i$,  take $A$  to be the appropriate
      power of $\lead p_i(\bar a)$.
  \end{itemize}
  As $AM_0 \subset M_1$, we have $F(AM_0) \subset F(p_i)$. 
Kernels of Wronskians are defined over the constants by Lemma~\ref{lem:wronsk_const}, so the kernel of the Wronskian of a tuple does not change if the tuple is multiplied by a nonzero element.
  Hence, $F(AM_0) = F(M_0) = F(\alpha(h))$, and so  $F(\alpha(h)) \subset F(p_i)$.
  Since $\alpha(h)(\bar a)=0$,  there are $r_1, \ldots, r_s \in  F(\alpha(h))\{\bar z\}=F(h)\{\bar z\} \subset 
  K\{\bar z\}$ and $\lambda_1, \ldots, \lambda_s \in K$ such that 
  \[
    \alpha(h) = \lambda_1r_1 + \ldots + \lambda_s r_s \text{ and } r_1(\bar a) = \ldots = r_s(\bar a) = 0.
  \]
  Applying $\alpha^{-1}$, we get 
  \[
  h = \alpha^{-1}(\lambda_1) r_1 + \ldots + \alpha^{-1}(\lambda_s) r_s,
  \]
 so $h(\bar a) = 0$, which is impossible, hence the  contradiction.
\end{proof}
\subsection{Multi-experiment identifiability}\label{sec:pfmain2}

\begin{proof}[Proof of Theorem~\ref{thm:me}]
For simplicity of notation, we  denote the tuple of variables $\bar{y},\bar{u}$ by $\bar{w}$.
Note that, for every $N \geqslant 1$, 
the set \[\bar{p}(\bar{w}_1),\ldots,\bar{p}(\bar{w}_N) \subset k\{\bar{w}_1,\ldots,\bar{w}_N\}\]
is a set of IO-equations of $\Sigma_N$. The coefficients of $\bar{p}(\bar{w}_1),\ldots,\bar{p}(\bar{w}_N)$  are also $c_{1,1},\ldots,c_{m,s_m}$.
Hence, as in \cite[Corollary~1 and Theorem~1]{ident-compare},
the field of $N$-experiment identifiable functions is contained in $\mathbb{C}(c_{1,1},\ldots,c_{m,s_m})$.

For the reverse inclusion, let $p \in \bar{p}$, 
\[
    p = \sum_{i=1}^sb_if_i + f_{s+1},
\]
where, for each $i$, $f_i \in \mathbb{C}\{\bar{w}\}$  and $f_1,\ldots,f_s$ are linearly independent over $\mathbb{C}$. 
By dividing $p$ by an element of $k$, we may assume that $\deg f_{s+1}=\deg p$. Let
\[
A :=
\begin{pmatrix}
f_1(\bar{a}_1)&\ldots&f_s(\bar{a}_1)\\
\vdots&\ddots&\vdots\\
f_1(\bar{a}_s)&\ldots&f_s(\bar{a}_s)
\end{pmatrix},
\]
where, for each $i$, $\bar a_i$ is the image of $\bar{w}_i$ modulo $I_{\Sigma_N}$.
We will first show that $\det A \ne 0$.
For this, let $M$ be a minimal (by size) zero minor of $A$. 
 Let, for some $i$ and $\ell$, $f_i(\bar{a}_\ell)$ appear in $M$ and  $q\in k\{\bar{w}\}$ be the differential polynomial obtained from $M$ by replacing  $f_j(\bar{a}_\ell)$ with $f_j(\bar{w})$, $1\leqslant j\leqslant s$. By the minimality of $M$  and linear independence of $f_1,\ldots,f_s$, $q(\bar{w})\ne 0$.
Since $q(\bar{a}_\ell)=0$, there exist $q_{i,j}\in k\{\bar{w}_j\}$ such that \begin{gather*}\forall i\, q_{i,1}(\bar{w}) \in I_\Sigma\text{ or }\ldots\text{ or }q_{i,s}(\bar{w})\in I_\Sigma\quad\text{and}\\ q=\sum_iq_{i,\ell}(\bar{w})\cdot\prod_{
{j=1}\atop{j\ne \ell}}^s q_{i,j}(\bar{a}_j).
 \end{gather*}
 Hence, there exist $\alpha$, and $q_1,\ldots,q_\alpha\in I_\Sigma$, and $b_1,\ldots,b_\alpha \in k\langle \bar{a}_1,\ldots,\bar{a}_{\ell-1},\bar{a}_{\ell+1},\ldots,\bar{a}_s\rangle$  such that $q = \sum\limits_{i=1}^\alpha b_iq_i$ and, for each $i$,
every monomial that appears in $q_i$ also appears in $q$  (and, therefore, in $p$).
Let $\tilde q$ be the primitive part of $q_{1}$ considered as a polynomial in its leader.  
Since $I_\Sigma$ is prime, $\tilde q \in I_\Sigma$. 
Since $\bar{p}$ is autoreduced and $\tilde q$ divides a linear combination of the monomials of $p$, the characteristic set $\tilde{\bar{p}}$ of  $\bar{p}\setminus \{p\} \cup \{\tilde q\}$ satisfies $\rank\tilde {\bar{p}} \leqslant \rank\bar{p}$. 
Hence, $\tilde{\bar{p}}$ is a characteristic set of $J$, and so \[\tilde{\bar{p}} = \bar{p}\setminus \{p\} \cup \{\tilde q\}.\] Thus, $\tilde{\bar{p}}$ is a characteristic presentation of $I_\Sigma$.  
If $\tilde q \ne q$, then $\deg\tilde q<\deg q$. If $\tilde q = q$, then $\tilde q$ has fewer monomials than $p$ does. 
Thus, in either case, $p/\tilde q \notin k$.
However, \cite[Theorem~3]{Boulier2000} implies that $p/\tilde{q} \in k$, which is a contradiction. 
This shows that $\det A\ne 0$.  Thus, the rows of $A$ are linearly independent.

Let $r = \rank \Wrn(f_1(\bar{a}),\ldots,f_s(\bar{a}))$ and the rows  $i_1=0,i_2,\ldots,i_r$ of the Wronskian be linearly independent.  
 Since the rows of $A$ form a basis of $\mathbb{C}\langle \bar{a}_1,\ldots,\bar{a}_s\rangle^s$, there exist rows $j_1,\ldots,j_{s-r}$ of $A$ such that they together with the rows $i_1,\ldots,i_r$ of the Wronskian form a basis of $\mathbb{C}\langle \bar{a}_1,\ldots,\bar{a}_s\rangle^s$ as well. 
 Hence,
 \[B := \begin{pmatrix}
f_1(\bar{a}_1)&\ldots&f_s(\bar{a}_1)\\
f_1^{(i_2)}(\bar{a}_1)&\ldots&f_s^{(i_2)}(\bar{a}_1)\\
\vdots&\ddots&\vdots\\
f_1^{(i_r)}(\bar{a}_1)&\ldots&f_s^{(i_r)}(\bar{a}_1)\\
f_1(\bar{a}_{j_1})&\ldots&f_s(\bar{a}_{j_1})\\
\vdots&\ddots&\vdots\\
f_1(\bar{a}_{j_{s-r}})&\ldots&f_s(\bar{a}_{j_{s-r}})\end{pmatrix},\]
is invertible. 
Replacing $\bar{a}_1,\bar{a}_{j_1},\ldots,\bar{a}_{j_{s-r}}$ in $\det B$ by the indeterminates $\bar{w}_1,\ldots,\bar{w}_{s-r+1}$, we obtain a differential polynomial with coefficients in $\mathbb{C}$ that does not belong to he vanishing ideal of $\bar{a}_1,\bar{a}_{j_1},\ldots,\bar{a}_{j_{s-r}}$. Since this ideal is the same as the vanishing ideal of $\bar{a}_1,\bar{a}_2,\ldots,\bar{a}_{s-r+1}$, we conclude that the matrix
 \[C :=  \begin{pmatrix}
f_1(\bar{a}_1)&\ldots&f_s(\bar{a}_1)\\
f_1^{(i_2)}(\bar{a}_1)&\ldots&f_s^{(i_2)}(\bar{a}_1)\\
\vdots&\ddots&\vdots\\
f_1^{(i_r)}(\bar{a}_1)&\ldots&f_s^{(i_r)}(\bar{a}_1)\\
f_1(\bar{a}_2)&\ldots&f_s(\bar{a}_2)\\
\vdots&\ddots&\vdots\\
f_1(\bar{a}_{s-r+1})&\ldots&f_s(\bar{a}_{s-r+1})\end{pmatrix}\]
is invertible.
Thus,
\[\begin{pmatrix}
b_1\\
\vdots\\
b_s
\end{pmatrix} = C^{-1}
\begin{pmatrix}f_{s+1}(\bar{a}_1)\\
f_{s+1}^{(i_2)}(\bar{a}_1)\\
\vdots\\
f_{s+1}^{(i_r)}(\bar{a}_1)\\
f_{s+1}(\bar{a}_2)\\
\vdots\\
f_{s+1}(\bar{a}_s)
\end{pmatrix},\]
which is in $\mathbb{C}\langle\bar{a}_1,\ldots,\bar{a}_{s-r+1}\rangle^s$ and so is $(s-r+1)$-experiment identifiable. Thus, the field of IO-identifiable functions is $N$-experiment identifiable.
\end{proof}

\section{Intersection for rational function fields}\label{sec:fi}

In this section, we will describe
how \cite[Algorithm~2.38]{Binder2009} can be used to compute the intersection $L_1\cap L_2$, where $L_1=\mathbb{C}(\bar\mu)$ and $L_2 = F(\bar p)$, as required by Algorithm~\ref{alg:main}.

Algorithm~\ref{alg:field_intersection} below is a version of~\cite[Algorithm~2.38]{Binder2009}.
It was shown~\cite[p.~37-38]{Binder2009} that the output of the algorithm is correct if the algorithm terminates.
Termination was proved only if both input fields are algebraically closed in the ambient rational function field.
To use the algorithm in our applications, we relax this condition to requiring only one of the fields to be algebraically closed ($\mathbb{C}(\bar{\mu})$ in $\mathbb{C}(\bar{\mu}, \bar{x})$) in Proposition~\ref{prop:termination}.

\begin{algorithm}[H]
\caption{Intersection of fields (a version of  \cite[Algorithm~2.38]{Binder2009})}\label{alg:field_intersection}
\begin{description}
\item[Input] Tuples $\bar{f} := (f_1, \ldots, f_s)$ and $\bar{g} := (g_1, \ldots, g_\ell)$ such that $f_1, \ldots, f_s, g_1, \ldots, g_\ell \in K(\bar{x})$, where $\bar{x} := (x_1, \ldots, x_n)$;
\item[Output] If terminates, returns generators of $K(\bar{f}) \cap K(\bar{g})$.
\end{description}

\textbf{Notation:} Introduce new variables $\overline{X} := (X_1,\ldots, X_n)$.
In the algorithm, for a set $S \subset K(\bar{x})[\overline{X}]$, $\langle S\rangle$ will denote the ideal generated by $S$ in $K(\bar{x})[\overline{X}]$.

\begin{enumerate}[label = \textbf{(Step~\arabic*)}, leftmargin=7.5mm, align=left,itemsep=0in]
    \item For every $1 \leqslant i \leqslant s$, write $f_i(\bar{x}) = \frac{n_i(\bar{x})}{d_i(\bar{x})}$ so that $n_i, d_i \in K[\bar{x}]$, and set $D(\bar{x}) := d_1\cdot\ldots\cdot d_s$;    
    \item Set $i := 1$, $I_1 := \langle 1 \rangle$ and 
    \[
      J_1 := \big\langle n_1(\overline{X}) - f_1(\bar{x})d_1(\overline{X}), \ldots, n_s(\overline{X}) - f_s(\bar{x})d_s(\overline{X}) \big\rangle \colon D(\overline{X})^\infty;
    \]
    
    \item While $I_i \neq J_i$ do
    \begin{enumerate}
        \item $I_{i + 1} := \langle J_i \cap K(\bar{g})[\overline{X}] \rangle$;
        \item $J_{i + 1} := \langle I_{i + 1} \cap K(\bar{f})[\overline{X}] \rangle$;
        \item $i := i + 1$;
    \end{enumerate}
    \item Compute any reduced Gr\"obner basis of $J_i$ and return its coefficients.
\end{enumerate}
\end{algorithm}

\begin{prop}[Termination of Algorithm~\ref{alg:field_intersection}]\label{prop:termination}
  In the notation of Algorithm~\ref{alg:field_intersection}, if $K(\bar{f})$  is algebraically closed in $K(\bar{x})$ or  $K(\bar{g})$ is algebraically closed in $K(\bar{x})$, then Algorithm~\ref{alg:field_intersection} terminates.
\end{prop}

\begin{lem}\label{lem:intersection_commutes}
  Let $I_0, I_1, \ldots ,I_s \subset K[\bar{x}]$, where $\bar{x} = (x_1, \ldots, x_n)$, be ideals such that $I_0 = I_1 \cap \ldots \cap I_s$, and let $L \subset K$ be a subfield.
  For $0 \leqslant j \leqslant s$, we define $J_j$ to be the ideal in $K[\bar{x}]$ generated by $I_j \cap L[\bar{x}]$.
  Then $J_0 = J_1 \cap \ldots \cap J_s$.
\end{lem}

\begin{proof}
  Since $I_0 \cap L[\bar{x}] = (I_1 \cap L[\bar{x}]) \cap \ldots \cap (I_s \cap L[\bar{x}])$, we have $J_0 \subset J_1 \cap \ldots \cap J_s$.
  
  Now we prove the reverse inclusion.
  Let $\{a_{\lambda}\}_{\lambda \in \Lambda}$ be an $L$-basis of $K$.
  Consider $b \in J_1 \cap \ldots \cap J_s$.
  We write $b = \sum\limits_{\lambda \in \Lambda} b_{\lambda}a_{\lambda}$, where $b_{\lambda} \in L[\bar{x}]$ for every $\lambda \in \Lambda$ and only finitely many of them are not zeroes.
  Consider any $1 \leqslant j \leqslant s$. 
  Since $J_j$ has a set of generators with coefficients in $L$, the inclusion $b \in J_j$ implies that $b_{\lambda} \in I_j \cap L[\bar{x}]$ for every $\lambda\in\Lambda$.
  Therefore, $b_{\lambda} \in I_0 \cap L[\bar{x}]$ for every $\lambda \in \Lambda$. Thus, $b \in J_0$.
\end{proof}

\begin{proof}[Proof of Proposition~\ref{prop:termination}]
We will assume that $K(\bar{f})$ is algebraically closed in $K(\bar{x})$.
The proof for the case of $K(\bar{g})$ being algebraically closed in $K(\bar x)$ is analogous.
Assume that the algorithm does not terminate.
By the construction, $I_j \supset J_j$ for every $j \geqslant 1$.
The ideals $I_1$ and $J_1$ are radical (the latter is due to~\cite[Definition~2.16 and Proposition~2.21]{Binder2009} and since the intersection of a radical ideal with a subring is radical and the extension of a radical ideal is radical). It follows then that all $I_i$'s and $J_i$'s are radical.
For every $i \geqslant 1$, we define $d_i$ to be the minimum of the dimensions of the prime components $P$ of $J_i$ such that $P \not\supset I_i$.
We will show that the sequence $d_i$ is strictly increasing thus arriving at a contradiction.

Fix $i \geqslant 1$. 
Let $P_1, \ldots, P_m$ be the prime components of $J_i$ so that $P_1, \ldots, P_r$ are the components of the dimension $< d_i$ and $P_{r + 1}, \ldots, P_m$ are the components of the dimension $\geqslant d_i$. 
By the construction, $J_i$ is defined over $K(\bar{f})$.
\cite[Proposition~2.37]{Binder2009} implies that $P_1, \ldots, P_m$ are also defined over $K(\bar{f})$.

Since $I_i \supset J_i$, and $P_1, \ldots, P_r$ contain $I_i$, $P_1, \ldots, P_r$ are exactly the prime components of $I_i$ of dimension $< d$, so $Q := P_1 \cap \ldots \cap P_r$ is the intersection of the equidimensional components of $I_i$ of dimensions $< d$.
Therefore, since $I_i$ is defined over $K(\bar{g})$, $Q$ is defined over $K(\bar{g})$.
Hence, 
\begin{equation}\label{eq:eq_for_Q}
Q = \langle Q \cap K(\bar{g})[\overline{X}]\rangle = \langle Q \cap K(\bar{f})[\overline{X}]\rangle \supset I_{i + 1}.
\end{equation}
Consider 
\[
\mathcal{C} := \{C \mid C \text{ is a prime component of } \langle P_j \cap K(\bar{g})[\overline{X}]\rangle \text{ for } j > r\}
\]
\cite{Seidenberg75} implies that, for every $j > r$, all prime components of $\langle P_j \cap K(\bar{g})[\overline{X}]\rangle$ are of the same dimension, so, for all $C \in \mathcal{C}$, $\dim C \geqslant d_i$.
For every $C \in \mathcal{C}$, denote $C' := \langle C \cap K(\bar{f})[\overline{X}] \rangle$.
\cite[Proposition~2.37]{Binder2009} implies that $C'$ is prime.
If $C \neq C'$, then $\dim C' > d_i$.
Otherwise, $C' = C \supset I_{i + 1}$.
Therefore, due to Lemma~\ref{lem:intersection_commutes}, we have:
\[
  J_{i + 1} = \langle I_{i + i} \cap K(\bar{f})[\overline{X}]\rangle = \underbrace{\left(\langle Q \cap K(\bar{f})[\overline{X}]\rangle \cap \bigcap\limits_{C \in \mathcal{C}, C = C'} C'\right)}_{=: A} \cap \underbrace{\left( \bigcap\limits_{C \in \mathcal{C}, C \neq C'} C'\right)}_{=: B}
\]
Since $\langle Q \cap K(\bar{f})[\overline{X}]\rangle \supset I_{i + 1}$ (see~\eqref{eq:eq_for_Q}), we have $A \supset I_{i + 1}$.
Since every component of $B$ has dimension at least $d_i + 1$, we deduce that $d_{i + 1} > d_i$.
\end{proof}

%%%%%%%%%%%%%%%%%%%%%%%%%%%%%%%%%%%%%%%%%%%%%%%%%%

\section{Mathematical discussion for Theorems~\ref{prop:charset} and~\ref{thm:me}}

\begin{exmp}[Ranking dependency of $F(\bar{p})$ in Theorem~\ref{prop:charset}]\label{ex:dep}
We show that the field $F(\bar{p})$ from Theorem~\ref{prop:charset} can depend on the ranking although $\C(\bar{\mu}) \cap F(\bar{p})$ cannot.
Consider the following input-output equations
\[
  p_1 := y_1^2+y_2^2+y_3,\;\; p_2 := y_2'-1,\;\;  p_3 := y_3'-1.
\]
For the elimination differential ranking $y_1>y_2>y_3$, $p_1,p_2,p_3$ is the characteristic presentation of the prime differential ideal $P := \sqrt{[ p_1,p_2,p_3]}$. A calculation in {\sc Maple} shows that $F(p_1)=F(p_2)=F(p_3)=\mathbb{C}$, and so $F(\bar p)=\mathbb{C}$. However, a calculation in {\sc Maple} shows that $\bar q := \{q_1,q_2,q_3\}$, 
\begin{align*}
q_1&:= 2y_2 + 2y_1y_1' + 1,\\
q_2&:= 
4y_1^2y_1'^2 + 4y_1y_1' + 4y_1^2 + 4y_3 + 1,\\
q_3 &:= y_3' - 1,
\end{align*}
is the characteristic presentation of $P$ with respect to the elimination differential ranking $y_2>y_1>y_3$ and that $F(q_2)=\mathbb{C}(y_1y_1'+y_3)$ and $F(q_1)=F(q_3)=\mathbb{C}$, and so $F(\bar q)\supsetneq F(\bar p)$.
\end{exmp}

\begin{exmp}[Achieving the bound in Theorem~\ref{thm:me}]
A natural mathematical question about a bound is whether it is tight in the sense that the equality can be reached for all the values of the parameters appearing in the bound.
We will give an indication of the tightness of the bound from Theorem~\ref{thm:me} by providing, for every positive integers $h \leqslant n$, a model with $n + 1$ monomials in the IO-equations and the corresponding Wronskian having rank $h$ so that every element of the field of IO-identifiable functions is $(n - h + 1)$-identifiable but not necessarily $(n - h)$-identifiable.
Fix $h \leqslant n$ and consider the system
\begin{equation}\label{eq:sharp}
\Sigma=\begin{cases}
x_1'=c_1 + \sum\limits_{i=2}^n c_ix_i,\\
x_i^{(h)}=0,& 2\leqslant i \leqslant h\\
x_i'=0,&h+1\leqslant i\leqslant n\\
y_i=x_i,&1\leqslant i \leqslant n\end{cases}
\end{equation}
with unknown parameters $\{c_i, 1\leqslant i \leqslant n\}$.
By a calculation, \[\bar{p} = \Big\{y_1'-c_1 - \sum\limits_{i=2}^n c_iy_i,\,  
y_i^{(h)},\, 2\leqslant i \leqslant h,\,  
y_i',\, h+1\leqslant i\leqslant n\Big\}\]
is a set of IO-equations of~\eqref{eq:sharp}. We have modulo $I_\Sigma$:
\[
\Wrn(y_2,\ldots,y_n,1)=\begin{pmatrix}y_2&\ldots&y_h&y_{h+1}&\ldots&y_n&1\\
y_2'&\ldots&y_h'&0&\ldots&0&0\\
\vdots&\ddots&\vdots&\vdots&\ddots&\vdots&\vdots\\
y_2^{(h-1)}&\ldots&y_h^{(h-1)}&0&\ldots&0&0\\
0&\ldots&0&0&\ldots&0&0\\
\vdots&\ddots&\vdots&\vdots&\ddots&\vdots&\vdots\\
0&\ldots&0&0&\ldots&0&0
\end{pmatrix},\]
whose rank is $r_1$. On the one hand $r_1\leqslant h$ because the matrix has only $h$ non-zero rows. On the other hand, $\det\Wrn(y_2,\ldots,y_h,1) \notin I_\Sigma$ since $\Wrn(y_2,\ldots,y_h,1)$ is not reducible (to zero) by $\bar{p}$. Thus, $r_1=h$. Also, $s_1=n$ and, for all $i\geqslant 2$, $s_i =0$. So, by Theorem~\ref{thm:me}, for all \[N \geqslant s_1-r_1+1=n-h+1,\]
the field of IO-identifiable functions
$\mathbb{C}(c_1,\ldots,c_n)$
is $N$-experiment identifiable. We will show that it is not $(n-h)$-experiment identifiable, thus showing the desired tightness of the bound in Theorem~\ref{thm:me}. For this, consider the following set of IO-equations for the $(n-h)$-experiment system $\Sigma_{n-h}$:
\[
\bigcup_{j=1}^{n-h}\left\{y_{j,1}'-c_1 - \sum\limits_{i=2}^n c_iy_{j,i},\ \ 
{y_{j,i}^{(h)},\, 2\leqslant i \leqslant h\atop  
y_{j,i}',\, h+1\leqslant i\leqslant n}\right\}
\]
Let $a_{j,i}$ denote the image of $y_{j,i}$ modulo $I_{\Sigma_{n-h}}$.
Since, for all $i$ and $j$, $h+1\leqslant i\leqslant n$, $1\leqslant j\leqslant n-h$, $a_{j,i}$ is constant, we can define a differential field automorphism $\varphi$ of $\mathbb{C}\langle \bar{a}_1,\ldots,\bar{a}_{n-h}\rangle (c_1,\ldots,c_n)$ over $\mathbb{C}\langle \bar{a}_1,\ldots,\bar{a}_{n-h}\rangle$ by
\[\varphi(c_1,c_2,\ldots,c_h,c_{h+1},\ldots, c_n) := (c_1+b_{n-h+1},c_2\ldots,c_h,c_{h+1}+b_1,\ldots,c_n+b_{n-h}),\]
where
$(b_1,\ldots,b_{n-h+1}) \in \mathbb{C}\langle a_{j,i}\mid 1\leqslant j \leqslant n-h,\, h+1\leqslant i\leqslant n\rangle$ is a non-zero linear dependence among the columns of the $(n-h)\times(n-h+1)$ matrix
\[
\begin{pmatrix}
a_{1,h+1}&\ldots&a_{1,n}&1\\
\vdots&\ddots&\vdots&\vdots\\
a_{n-h,h+1}&\ldots&a_{n-h,n}&1
\end{pmatrix}.\]
Thus, there exists $i \in \{1,h+1,\ldots,n\}$ such that $c_i \notin \mathbb{C}\langle \bar{a}_1,\ldots,\bar{a}_{n-h}\rangle$, and so the IO-identifiable parameter $c_i$ is not $(n-h)$-experiment identifiable.
\end{exmp}

%%%%%%%%%%%%%%%%%%%%%%%%%%%%%%%%%%%%%%%%%%%%%%%%%%%%%%%%%%

\section{Computing an optimal representation~\eqref{eq:p_decomposition} in Theorem~\ref{thm:me}}\label{app:rep5}

In this section, we prove Lemma~\ref{lem:rank_factor_optimal} providing a sufficient condition for a decomposition~\eqref{eq:p_decomposition} to yield the optimal (compared to other decompositions) bound in Theorem~\ref{thm:me}.
Then we give Algorithm~\ref{alg:poly_decomposition} to compute such a decomposition, which basically computes an LU-decomposition of a matrix in the language of polynomials (see the proof of Lemma~\ref{lem:rank_factor_optimal}).

\begin{lem}\label{lem:rank_factor_optimal}
  Let $p(\bar{z}) \in \C(\bar{\mu})\{\bar{z}\}$ be a differential polynomial over a constant field $\C(\bar{\mu})$ in $\bar{z} = (z_1, \ldots, z_n)$, where
  $\bar{\mu} = (\mu_1, \ldots, \mu_m)$ are transcendental constants.
  Let $I \subset \C(\bar{\mu})\{\bar{z}\}$ be a prime differential ideal containing $p$.
  Consider two representations of $p$
  \[
      p = f_{s + 1} + \sum\limits_{j = 1}^{s} c_{j} f_{j}\;\;\text{ and }\;\;p = \tilde{f}_{\tilde{s} + 1} + \sum\limits_{j = 1}^{\tilde{s}} \tilde{c}_{j} \tilde{f}_{j}
  \]
  such that $f_1, \ldots, f_{s + 1}, \tilde{f}_1, \ldots, \tilde{f}_{\tilde{s} + 1} \in \C\{\bar{z}\}$, $c_1, \ldots, c_s, \tilde{c}_1, \ldots, \tilde{c}_{\tilde{s}} \in \C(\bar{\mu})$, $f_1, \ldots, f_{s + 1}$ are $\C$-linearly independent, and $1, c_1, \ldots, c_s$ are $\C$-linearly independent.
  We define $r$ and $\tilde{r}$ to be the ranks of $\Wrn(f_1, \ldots, f_s)$ and  $\Wrn(\tilde{f}_1, \ldots, \tilde{f}_{\tilde{s}})$ modulo $I$, respectively.
  Then
  \[
    s - r \leqslant \tilde{s} - \tilde{r} \quad\text{ and }\quad s \leqslant \tilde{s}.
  \]
\end{lem}

\begin{proof}
  Viewing $\C(\bar{\mu})\{\bar{z}\}$ as a tensor product of $\C$-vector spaces $\C(\bar{\mu}) \otimes_\C \C\{\bar{z}\}$, we can consider $p$ as an element of this tensor product.
  Then the linear independence of $1, c_1, \ldots, c_s$ and of $f_1, \ldots, f_{s + 1}$ implies that 
  \[
    1\otimes f_{s + 1} + c_1\otimes f_1 + \ldots + c_s\otimes f_s
  \]
  is a full-rank factorization of $p$~\cite[Theorem~3.13]{Stewart}.
  Since $1\otimes \tilde{f}_{\tilde{s} + 1} + \tilde{c}_1\otimes \tilde{f}_1 + \ldots + \tilde{c}_{\tilde{s}}\otimes \tilde{f}_{\tilde{s}}$ is another rank-one factorization of the same tensor, the proof of~\cite[Theorem~3.13]{Stewart} implies that $1, c_1, \ldots, c_s$ belong to the $\C$-span of $1, \tilde{c}_1, \ldots, \tilde{c}_{\tilde{s}}$ and $f_1, \ldots, f_{s + 1}$ belong to the $\C$-span of $\tilde{f}_1, \ldots, \tilde{f}_{\tilde{s} + 1}$.
  The former inclusion implies $s \leqslant \tilde{s}$.
  The latter implies that there exists a full-rank $\C$-matrix $M$ such that $(f_1, \ldots, f_{s + 1}) = (\tilde{f}_1, \ldots, \tilde{f}_{\tilde{s} + 1})M$.
  Therefore, any  nontrivial linear dependence of the images of $f_1, \ldots, f_{s + 1}$ in $\C(\bar{\mu})\{\bar{z}\} / I$ over the constants of the fraction field of $\C(\bar{\mu})\{\bar{z}\} / I$ yields (after multiplying by $M$) such a relation for the images of $\tilde{f}_1, \ldots, \tilde{f}_{\tilde{s} + 1}$.
  Therefore, the proof of Lemma~\ref{lem:wronsk_const} implies that the corank of $\Wrn(f_1, \ldots, f_{s})$ modulo $I$, $r - s$, does not exceed the corank of $\Wrn(\tilde{f}_1, \ldots, \tilde{f}_{\tilde{s}})$ modulo $I$, $\tilde{r} - \tilde{s}$.
\end{proof}

%%%%%%%%%%%%%%%%%%%%%%%%%%%%%

\begin{algorithm}[H]
\caption{Computing optimal representation for Theorem~\ref{thm:me}}\label{alg:poly_decomposition}
\begin{description}
\item[Input] a monic polynomial $p(\bar{x}) \in \C(\bar{\mu})[\bar{x}]$ (see Definition~\ref{def:char_pres}), where $\bar{x} = (x_1, \ldots, x_n)$ and $\bar{\mu} = (\mu_1, \ldots, \mu_m)$ are independent indeterminates;
\item[Output] A representation of $p$ of the form
\[
  p = f_{s + 1} + \sum\limits_{j = 1}^{s} c_{j} f_{j}
\]
in which $f_1, \ldots, f_{s + 1} \in \C[\bar{x}]$ are $\C$-linearly independent and $1, c_1, \ldots, c_s \in \C(\bar{\mu})$ are $\C$-linearly independent.
\end{description}

Fix an arbitrary ordering on the monomials in $\bar{\mu}$.
The leading monomial and leading coefficient of a polynomial $f$ w.r.t. this ordering will be denoted by $\operatorname{lm} f$ and $\operatorname{lc} f$, respectively.

\begin{enumerate}[label = \textbf{(Step~\arabic*)}, leftmargin=7.5mm, align=left,itemsep=0in]
    \item Compute the LCM $q(\bar{\mu}) \in \C[\bar{\mu}]$ of the denominators of the coefficients of $p$. 
    Set $P(\bar{\mu}, \bar{x}) := q\cdot p \in \C[\bar{\mu}, \bar{x}]$.
    \item Write $P$ as $\sum\limits_{i = 1}^\ell C_i(\bar{\mu}) M_i(\bar{x})$, where $M_1, \ldots, M_\ell$ are distinct monomials in $\bar{x}$ and $C_1, \ldots, C_\ell \in \C[\bar{\mu}]$, and $C_1 = q$ (possible since $p$ is monic).
    \item Let $S$ be a list of pairs from $\C[\bar{\mu}] \times \C[\bar{x}]$ initialized to be empty.
    \item\label{step:S_loop} For every $i = 1, \ldots, \ell$, do
    \begin{enumerate}
        \item for every $(A, B) \in S$, where $A \in \C[\bar{\mu}], B \in \C[\bar{x}]$
        \begin{equation}\label{eq:update_S}
            C_i := C_i - \frac{c}{\operatorname{lc}(A)} A, \quad B := B + \frac{c}{\operatorname{lc}(A)} M_i,
        \end{equation}
        where $c$ is the coefficient in front of $\operatorname{lm}(A)$ in $C_i$.
        \item\label{step:CM_append} if $C_i \neq 0$, append $(C_i, M_i)$ to $S$.
    \end{enumerate}
    \item Let $S = [(A_0, B_0), (A_1, B_1), \ldots, (A_{s}, B_{s})]$.
    Return $f_{s + 1} = B_0$ and $f_i = B_i$ and $c_i = \frac{A_i}{q}$ for every $1 \leqslant i \leqslant s$.
\end{enumerate}
\end{algorithm}

\begin{lem}
  Algorithm~\ref{alg:poly_decomposition} is correct.
\end{lem}

\begin{proof}
  First we will show that $p = \sum\limits_{i = 1}^s \frac{A_i}{q}B_i + B_0$.
  By the construction, we will have $A_0 = Q$, so this is equivalent to proving $p = \sum\limits_{i = 0}^s \frac{A_i}{q}B_i$.
  To prove this, we observe that the transformation~\eqref{eq:update_S} preserves the value
  \[
    \sum\limits_{(A, B) \in S} A\cdot B + C_i M_i.
  \]
  Therefore, after the $i$-th iteration of the loop in~\ref{step:S_loop}, the value $\sum\limits_{(A, B) \in S} AB$ is increased by $C_iM_i$.
  Since it starts with zero, it will be equal to $\sum\limits_{j = 1}^\ell C_jM_j = P$ after~\ref{step:S_loop}.
  Therefore, $\sum\limits_{i = 0}^s \frac{A_i}{q}B_i = \frac{P}{q} = p$.
  
  To prove the $\C$-linear independence of $B_j$'s, for each $1 \leqslant j \leqslant s$, consider the pair $(C_i, M_i)$ that was the $j$-th appended pair for~\ref{step:CM_append}.
  Then $M_i$ will not appear in any of $B_{j + 1}, \ldots, B_{s}$, so $B_1, \ldots, B_s$ are $\C$-linearly independent.
  
  The linear independence of $A_0, \ldots, A_s$ follows from the fact that, $\operatorname{lm}(A_j)$ does not appear in $A_{j + 1}, \ldots, A_{s}$ for every $0 \leqslant j \leqslant s$, and this property is due to the reduction procedure~\eqref{eq:update_S}.
\end{proof}

\end{document}